\documentclass[prl,aps,twocolumn,superscriptaddress]{revtex4-1}
\usepackage[dvips]{graphicx}
\usepackage[british]{babel}
\usepackage{amsmath,amssymb,amsthm,dsfont,hyperref,ctable,url,braket,mathtools}
\usepackage[percent]{overpic}
\input{qcircuit}

\newcommand{\Tr}{\operatorname{Tr}}
\newcommand{\rank}{\operatorname{rank}}
\newcommand{\pds}[2]{\mathcal{S}_{#1}(#2)}
\newcommand{\N}[1]{\left|\!\left|{#1}\right|\!\right|}
\newcommand{\eig}[2]{\lambda_{#1}\left(#2\right)}

\DeclarePairedDelimiter{\floor}{\lfloor}{\rfloor}

\newtheorem{prop}{Proposition}
\newtheorem{thm}{Theorem}
\newtheorem{cor}{Corollary}

\begin{document}

\title{Device-independent tests of quantum measurements}

\author{Michele \surname{Dall'Arno}}

\email{cqtmda@nus.edu.sg}

\affiliation{Centre    for   Quantum    Technologies,   National
  University of Singapore, 3 Science Drive 2, 117543, Singapore}

\author{Sarah \surname{Brandsen}}

\email{sbrandse@caltech.edu}

\affiliation{Centre    for   Quantum    Technologies,   National
  University of Singapore, 3 Science Drive 2, 117543, Singapore}

\affiliation{California   Institute  of   Technology,  1200   E.
  California Blvd, Pasadena, CA 91125, United States}

\author{Francesco \surname{Buscemi}}

\email{buscemi@is.nagoya-u.ac.jp}

\affiliation{Graduate School of Informatics, Nagoya
  University, Chikusa-ku, Nagoya, 464-8601, Japan}

\author{Vlatko \surname{Vedral}}

\email{phyvv@nus.edu.sg}

\affiliation{Atomic  and  Laser Physics,  Clarendon  Laboratory,
  University  of  Oxford,  Parks  Road,  Oxford  OX13PU,  United
  Kingdom}

\affiliation{Centre    for   Quantum    Technologies,   National
  University of Singapore, 3 Science Drive 2, 117543, Singapore}

\date{\today}

\begin{abstract}
  We consider  the problem of  characterizing the set  of input-output
  correlations that can  be generated by an  arbitrarily given quantum
  measurement.   Our main  result is  to provide  a closed-form,  full
  characterization of  such a  set for any  qubit measurement,  and to
  discuss its geometrical interpretation.  As applications, we further
  specify  our results  to the  cases of  real and  complex symmetric,
  informationally complete measurements and mutually unbiased bases of
  a qubit, in the presence of isotropic noise. Our results provide the
  optimal device-independent tests of quantum measurements.
\end{abstract}

\maketitle

In operational quantum theory, it is a natural question
to ask  whether a  given data  sample, provided  in the
form   of   a  conditional   probability   distribution
representing the measured  input-output correlation, is
compatible  with  a  particular  hypothesis  about  the
theoretical   model  underlying   the  experiment.    A
theoretical  model can  be more  or less  specific: for
example, it could consist only of a general hypothesis
about the theory  describing the physics, as in the
case in  Bell tests~\cite{Bel64, CHSH69, Cir80},  or it
could be extremely detailed, as  in the case
of  a   tomographic  reconstruction~\cite{PCZ97,  CN97,
  DL00}.   More   generally,  a  hypothesis   could  be
specific  about a  \textit{portion}  of the  underlying
model, while leaving  the remaining elements completely
uncharacterized.

Here we  address hypotheses about the measurement producing the final outcomes of the experiment.  
This is  the problem of characterizing the set $\pds{}{\pi}$  of all input-output correlations
$p_{y|x}=\Tr[\rho_x\   \pi_y]$   compatible   with   an
arbitrarily given quantum measurement  $\pi := \{ \pi_y
\}$  (i.e., the  hypothesis)  and any  family of  input
quantum states $\{\rho_x \}$, namely,
\begin{align}
  \label{eq:meastest}
  p_{y|x}=\Tr[\rho_x\ \pi_y]
  \; \longleftrightarrow \;
  \begin{aligned}
    \Qcircuit @C=4pt  @R=4pt {  \push{x \;}  & \prepareC{\rho_x}
      \cw & \measureD{\pi_y} & \cw & \push{\; y} }
  \end{aligned} \; .
\end{align}
We first note  that   a  correlation  $p$  is
compatible with  measurement $\pi$  if and only  if for
any fixed $x$  there exists a state  $\rho_x$ such that
$p_{y|x} = \Tr[\rho_x \pi_y]$.  Hence, $\pds{}{\pi}$ is
fully characterized by the  range of $\pi$, namely, the
set  $\pds{1}{\pi}$  of  output  distributions  $q_y  =
\Tr[\rho\ \pi_y]$ generated by  $\pi$ for varying input
state  $\rho$.   This is  in  stark  contrast with  the
analogous   problem  of   characterizing  the   set  of
correlations compatible  with a given  quantum channel,
which    in   general    requires    more   than    one
input~\cite{DBB16}.

Our main  result is  a closed-form  characterization of
the  range   $\pds{1}{\pi}$,  and  hence  of   the  set
$\pds{}{\pi}$ of all  compatible correlations, when the
hypothesis $\pi$ is a  qubit measurement.  It turns out
that   an   output   distribution  $q_y$   belongs   to
$\pds{1}{\pi}$ if and only if
\begin{align}
  \label{eq:range}
  \begin{cases}
    (\openone - Q^+ Q) (q - t) = 0, \\
    (q - t)^T Q^+ (q - t) \le 1,
  \end{cases}
\end{align}
where  $\openone$ is  the identity  matrix, $t$  is the
vector $t_y  := \frac12 \Tr[\pi_y]$, $Q$  is the matrix
$Q_{y_0,y_1}  :=  \frac12  \Tr[\pi_{y_0}  \pi_{y_1}]  -
\frac14 \Tr[\pi_{y_0}] \Tr[\pi_{y_1}]$, and $(\cdot)^+$
represents               the              Moore-Penrose
pseudoinverse~\cite{AG03}.   Imposing  the   system  of
equalities   in   Eq.~\eqref{eq:range}  causes   linear
dependencies --  if any  -- among  measurement elements
$\pi_y$  to   emerge  as  linear  constraints   on  the
probabilities.  Such linear dependencies are present in
any  overcomplete  measurement.   Provided  that  these
constraints   are   satisfied,    the   inequality   in
Eq.~\eqref{eq:range}    recasts    --    through    the
transformation  $Q^+$  --   the  set  of  distributions
compatible  with  $\pi$  as an  ellipsoid  centered  on
distribution $t$.  This in particular provides a simple
and clear  geometrical representation for the  range of
any qubit measurement.

As an  application of  our general results,  we further
simplify Eq.~(\ref{eq:range}) for some relevant classes
of qubit measurements,  i.e.  symmetric informationally
complete  (SIC) measurements~\cite{Zau99}  and mutually
unbiased bases~\cite{KR05} (MUB), both  in the real and
complex cases, and in  the presence of isotropic noise.
SIC  measurements play  a fundamental  role in  quantum
tomography~\cite{PCZ97,     CN97,    DL00},     quantum
communication~\cite{DDS11,  SS14, DBO14,  Szy14, Dal14,
  Dal15,   BDS16},    and   foundations    of   quantum
theory~\cite{FS03,  FS09,  FS11, AEF11,  Fuc12},  while
MUBs     are     pivotal    elements     in     quantum
cryptography~\cite{BB84},      entropic     uncertainty
relations~\cite{WW10,  BR11,  BHOW13}, and  locking  of
classical information in quantum states~\cite{DHLST04}.

Our  results  represent  a  further  step  towards  characterizing time-like correlations compatible with quantum theory through device-independent (DI) tests ~\cite{DBB16}.  The aim of DI tests is
that  of  falsifying  hypotheses about  the  underlying
physical  system, which  is  considered accessible  only
through     the     \textit{classical}     input-output
correlations it generates. In particular, no assumptions are made about the underlying physical system or how such correlations have been generated.    This
approach has been considered in previous literature for
the  problems   of  falsifying  hypotheses   about  the
dimension~\cite{GBHA10,  HGMBAT12,  ABCB11, DPGA12}  or
the average entropy~\cite{CBB15} of the input ensemble.

However,   while  here   we  provide   a  \textit{full}
characterization of  $\pds{}{\pi}$ (in  particular, for
the   qubit   case,    the   ellipsoid   described   in
Eq.~(\ref{eq:range})  can   be  \textit{plotted}),  the
application of previous results~\cite{GBHA10, HGMBAT12,
  ABCB11,   DPGA12,CBB15}   would    allow  one  to   probe
$\pds{}{\pi}$ along a fixed  radial direction only.  In
this sense, our results  provide {\em optimal} DI tests
of  quantum  measurements,   aimed  at  falsifying  the
hypothesis that an  observed input-output correlation $
p_{y|x} $  is generated by a  given quantum measurement
$\{ \pi_y \}$.

{\em Characterization  of $\pds{}{\pi}$.}  ---  We make
use   of  standard   results  in   quantum  information
theory~\cite{NC00}.  Any quantum  state $\rho$  is most
generally  described  by  a density  matrix,  namely  a
positive semidefinite unit-trace operator.  Any quantum
measurement $\pi$ is most generally described by a {\em
  positive operator-valued  measure} (POVM) $\pi  := \{
\pi_y  \}$,  namely  a  set  of  positive  semidefinite
operators such that $\sum_y \pi_y = \openone$.

For any  POVM $\pi  := \{ \pi_y  \}_{y=0}^{n-1}$, where
$n$   denotes  the   number   of   elements,  the   set
$\pds{}{\pi}$  of compatible  input-output correlations
is    formally     defined    as     $\pds{}{\pi}    :=
\bigcup_{m=1}^\infty        \pds{m}{\pi}$,        where
$\pds{m}{\pi}$   denotes   the    set   of   compatible
conditional probability distributions  $p := \{ p_{y|x}
\}$, upon  the input of  any set of $m$  unknown states
$\{ \rho_x \}$, that is
\begin{align*}
  \pds{m}{\pi} := \{ p \;  | \exists \; \{ \rho_x \}_{x
    =  0}^{m-1} \textrm{  s.t. }  p_{y|x} =  \Tr[\rho_x
  \pi_y ]\}.
\end{align*}

First,   we   notice   that,   for   any   fixed   $m$,
$\mathcal{S}_{m}(\pi)$ is  convex: indeed, for  any two
sets $\{ \rho_x \}$ and $\{ \sigma_x \}$ of $m$ states,
the conditional probability distribution $\{ p_{y|x} :=
\lambda  \Tr[\sigma_x \pi_y]  + (1-\lambda)  \Tr[\rho_x
\pi_y]  = \Tr[  \left( \lambda  \sigma_x +  (1-\lambda)
  \rho_x \right) \pi_y]  \}$ belongs to $\pds{m}{\pi}$,
since  $\left\{ \lambda  \sigma_x +  (1-\lambda) \rho_x
\right\}$ is itself a set of $m$ states.

Therefore,   as  a   consequence   of  the   hyperplane
separation  theorem~\cite{BV04},  it   is  possible  to
detect  any conditional  probability $p$  lying outside
the  set $\pds{m}{\pi}$  through  the  violation of  an
inequality    involving    a   linear    function    of
$p$~\cite{Bus12}.     More     explicitly,    $p    \in
\pds{m}{\pi}$ if and only if
\begin{align}
  \label{eq:compatibility}
  \max_w \left[ \Tr[w^T p] - W(\pi, w) \right] \le 0,
\end{align}
where the maximum is over any real matrix $w$ (with the
same  dimensions  as  $p$),   referred  to  as  a  {\em
  witness}, and $W(\pi, w)$  is defined as $\max_{q \in
  \pds{m}{\pi}}  \Tr[w^T q]$  and is  referred to  as a
{\em      witness     threshold}.       Notice     that
Eq.~\eqref{eq:compatibility}    corresponds    to    an
unconstrained maximin optimization problem.

As a preliminary remark,  let us discuss two properties
of Eq.~\eqref{eq:compatibility}  that will  be relevant
in  the  following. Since  $W(\pi,  w)$  is a  positive
homogeneous function,  i.e. $W(\pi, \alpha w)  = \alpha
W(\pi, w)$ for any $\alpha  \ge 0$, the {\em rescaling}
transformation $w  \to \alpha w$  for any $\alpha  > 0$
leaves      Eq.~\eqref{eq:compatibility}     invariant.
Moreover,  by   direct  computation  it   follows  that
Eq.~\eqref{eq:compatibility}  is  invariant  under  the
{\em  shifting}   transformation  $w  \to   w'$,  where
$w'_{x,y} := w_{x,y} + k_x$, for any $\vec{k}$.

Let us first solve the  optimization appearing in the definition
of $W(\pi, w)$. One has
\begin{align*}
  W(\pi, w) :=  & \sup_{\{ \rho_x \}}  \sum_x \Tr \left[\rho_x
    \left( \sum_y w_{x,y} \pi_y \right)  \right] \\ \le & \sum_x
  \sup_{\{  \rho_x \}}  \Tr \left[\rho_x  \left( \sum_y  w_{x,y}
      \pi_y \right) \right]
\end{align*}
where the  inequality is saturated  if and  only if $\{  \rho_x :=
\ket{\phi_x} \!\!  \bra{\phi_x}\}$, where $\{ \ket{\phi_x} \}$ are
the  eigenvectors  corresponding  to  the  largest  eigenvalue  of
$\sum_y w_{x,y}  \pi_y$.  In this  case one has  $\Tr \left[\rho_x
  \left( \sum_y w_{x,y} \pi_y  \right) \right] = \eig{\max}{\sum_y
  w_{x,y} \pi_y }$, where  $\eig{\max}{\cdot}$ denotes the largest
eigenvalue  of $(\cdot)$.   Notice that  $\eig{\max}{\cdot}$ is  a
convex function~\cite{BV04}.   Then our first  preliminary result immediately follows.

\begin{prop}[Witness threshold]
  \label{prop:threshold}
  For any POVM $\pi$ and  any witness $w$, the witness threshold
  $W(\pi, w)$ is given by
  \begin{align}
    \label{eq:threshold}
    W(\pi, w) = \sum_x \eig{\max}{\sum_y w_{x,y} \pi_y}.
  \end{align}
\end{prop}
Proposition~\ref{prop:threshold}  recasts  the  optimization  in
Eq.~\eqref{eq:compatibility}   as   an   unconstrained   concave
maximization problem over any witness $w$.

The transformation $w \to w'$, where $w'_{x,y} := \mu_x w_{x,y}$
with $\mu$ a probability distribution,  i.e. $\N{\mu}_1 = 1$ and
$\mu \ge 0$ (here and  in the following $\N{v}_p := \left(\sum_k
  v_k^p\right)^{1/p}$ denotes the $p$-norm  of vector $v$), maps
Eq.~\eqref{eq:compatibility} into
\begin{align*}
  \max_{\substack{\mu \ge 0 \\ \N{\mu}_1 = 1}} \max_w \sum_x
  \mu_x \left[  \sum_y p_{y|x} w_{x,y} -  \eig{\max}{\sum_y w_{x,y}
      \pi_y } \right] \le 0,
\end{align*}
where  the maximum  over probability  distributions $\mu$  is of
course attained when $\mu_x = \delta_{x,x^*}$, with
\begin{align*}
  x^*  = \arg  \max_x  \max_w \left[  \sum_y  p_{y|x} w_{x,y}  -
    \eig{\max}{\sum_y w_{x,y} \pi_y } \right].
\end{align*}
To  summarize, the  above  calculation  shows that  the
optimization   of  the   witness   $w$   can  be   done
independently  for each  $x$.  We  therefore obtain our second preliminary result.  

\begin{prop}
  \label{prop:meastest}
  For  any  given   POVM  $\pi  :=  \{   \pi_y  \}$,  a
  conditional   probability    distribution   $p_{y|x}$
  belongs  to $\pds{}{\pi}$  if  and only  if, for  any
  fixed $x  = \bar{x}$, $q_y  := \{ p_{y|\bar{x}}  \} $
  belongs to $\pds{1}{\pi}$.
\end{prop}

Due to Proposition~\ref{prop:meastest}, without loss of
generality  we   solve  the  optimization   problem  in
Eq.~\eqref{eq:compatibility} when  $m = 1$.   Hence, in
the following we will  consider output distributions $q
= \{ q_y \}$  rather than input-output distributions $p
= \{ p_{y|x} \}$.

{\em Range of qubit  measurements} --- In what follows,
we restrict  our attention  to the  case of  qubit POVM
$\pi$.    Let  $t   \in   \mathbb{R}^n$   and  $S   \in
\mathbb{R}^{n  \times   3}$  be  defined  by   $t_y  :=
\Tr[\pi_y]/2$ and  $S_{y,j} :=  \Tr[\pi_y \sigma_j]/2$,
where $\{  \sigma_j \}_{j=1}^3$ are the  Pauli matrices
$\{\sigma_1\equiv   X,\sigma_2\equiv   Y,\sigma_3\equiv
Z\}$  for some  fixed computational  basis, so  one has
$\pi_y = t_y \openone + \sum_{k=1}^3 S_{y,k} \sigma_k$.
Naturally,  our  result  will  be  independent  of  the
particular  choice of  the  computational basis.   With
this   parametrization,   the  witness   threshold   in
Eq.~\eqref{eq:threshold} becomes
\begin{align*}
  W(\pi, w) = t^T w + \N{S^T w}_2.
\end{align*}
Accordingly, Eq.~\eqref{eq:compatibility} becomes
\begin{align}
  \label{eq:qubit-compatibility-00}
  \max_w \left[ (q - t)^T w - \N{S^T w}_2 \right] \le 0.
\end{align}
For any $w$  such that $(p-t)^T w \neq  0$, let $\alpha
:= |(p-t)^T w|$ and $w'  := \alpha^{-1} w$, and let $w'
= w$  otherwise.  The transformation $w  \to w'$ leaves
Eq.~\eqref{eq:compatibility}       invariant.        So
Eq.~\eqref{eq:qubit-compatibility-00} becomes
\begin{align}
  \label{eq:qubit-compatibility-01}
  \max_{\substack{w \\ (q-t)^T w = \pm 1, 0}} \left[ (q
    - t)^T w - \N{S^T w}_2 \right] \le 0.
\end{align}
If   $(q  -   t)^T   w   =  -1,   0$,   one  has   that
Eq.~\eqref{eq:qubit-compatibility-01}    is   trivially
satisfied. Thus, we focus in  the following on the case
$(q      -       t)^T      w      =       1$,      when
Eq.~\eqref{eq:qubit-compatibility-01} becomes
\begin{align}
  \label{eq:qubit-compatibility-02}
  \min_{\substack{w  \\ (q-t)^T w  = 1}}  \N{S^T
    w}^2_2 \ge 1.
\end{align}

The                   optimization                   in
Eq.~\eqref{eq:qubit-compatibility-02}       is       an
equality-constrained  quadratic problem.   Its solution
leads to our second main result.

\begin{thm}[Range of qubit measurements]
  \label{thm:qubit-compatibility}
  An output  distribution $q :=  \{ q_y \}$  belongs to
  the   range  $\pds{1}{\pi}$   of   any  given   qubit
  measurement $\pi := \{ \pi_y \}$ if and only if
  \begin{align}
    \label{eq:qubit-compatibility}
    \begin{cases}
      (\openone - Q^+ Q) (q - t) = 0,\\
      (q - t)^T Q^+ (q - t) \le 1,
    \end{cases}
  \end{align}
  where $\openone$  is the identity matrix,  $t$ is the
  vector $t_y := \frac12 \Tr[\pi_y]$, $Q$ is the matrix
  $Q_{y_0, y_1}  := \frac12 \Tr[\pi_{y_0}  \pi_{y_1}] -
  \frac14    \Tr[\pi_{y_0}]     \Tr[\pi_{y_1}]$,    and
  $(\cdot)^+$     represents      the     Moore-Penrose
  pseudoinverse.
\end{thm}

Before           proceeding          to           prove
Theorem~\ref{thm:qubit-compatibility},  let us  discuss
the  role  of  matrix  $Q$ and  provide  a  geometrical
interpretation.

Remarkably,  matrix $Q$  has the  form of  a covariance
matrix and  quantifies the statistical overlap  of POVM
elements.        Indeed,       as       proved       in
Proposition~\ref{prop:corrmatrix}  in the  Supplemental
Material, element $Q_{y_0,y_1}$  is minimized and equal
to $-1/4$ if $\pi_{y_0}$  and $\pi_{y_1}$ are rank-one,
unit-trace,   orthogonal    effects.    Also,   element
$Q_{y_0,y_1}$ is maximized and equal to $1/4$ if $y_0 =
y_1$   and  $\pi_{y_0}$   is  a   rank-one,  unit-trace
effect.   Finally,  element   $Q_{y_0,y_1}   =  0$   if
$\pi_{y_0}$  and $\pi_{y_1}$  are both  proportional to
the identity operator.

Let  us first  focus  on the  system  of equalities  in
Eq.~\eqref{eq:qubit-compatibility}.   Denoting  by  $l$
the maximum number of  linearly independent elements in
$\{\pi_y \}$, the number of  equations in the system is
$n - l + 1$,  each identifying an $(n - 1)$-dimensional
hyperplane.   This   comes  from  the  fact   that,  by
definition, one  has $\rank Q^+Q  = l - 1$,  which also
implies  $\rank(\openone  -  Q^+  Q)  = n  -  l  +  1$.
Moreover,   when  all   POVM   elements  are   linearly
independent,  namely $n=l$,  the only  equation in  the
system  is  $\N{q}_1  = \N{t}_1=1$.   This  follows  by
explicit computation: in this case, $(\openone - Q^+Q)$
turns out to coincide with the rank-one projector along
the vector  with all unit entries.   Hence, in general,
the           system            of           equalities
in~(\ref{eq:qubit-compatibility})   represents   linear
dependencies among POVM elements $\{ \pi_y \}$.

Let    us   now    focus   on    the   inequality    in
Eq.~\eqref{eq:qubit-compatibility}, which represents an
$n$-dimensional degenerate  (hyper-) ellipsoid centered
on  probability $t$.   More  precisely, the  inequality
represents the Cartesian  product of $\mathbb{R}^{n-3}$
with     a      three-dimensional     ellipsoid,     or
$\mathbb{R}^{n-2}$ with  a two-dimensional  ellipse, or
$\mathbb{R}^{n-1}$  with   a  one-dimensional  segment,
depending  on  whether $l  =  4,  3, 2$,  respectively.
Accordingly,          the          solution          of
Eq.~\eqref{eq:qubit-compatibility} is  an ellipsoid, an
ellipse,  or  a   segment,  respectively,  embedded  in
$\mathbb{R}^n$.

We      now     turn      to      the     proof      of
Theorem~\ref{thm:qubit-compatibility}.

\begin{proof}
  By explicit computation,  it immediately follows that
  matrix $Q$ defined in the statement can be written in
  terms         of         matrix        $S$         in
  Eq.~\eqref{eq:qubit-compatibility-02}   as  $Q   =  S
  S^T$. Indeed,  using the  decomposition $\pi_y  = t_y
  \openone   +  \sum_{k=1}^3   S_{y,k}  \sigma_k$   one
  immediately has  $\frac12 \Tr[\pi_{y_0}  \pi_{y_1}] =
  t_{y_0} t_{y_1} + \sum_{k=1}^3 S_{y_0,k} S_{y_1, k}$.

  Equality-constrained quadratic problems can be solved
  explicitly~\cite{BV04,   AG03}.   In   the  case   of
  Eq.~\eqref{eq:qubit-compatibility-02} we thus have
  \begin{align}
    \label{eq:quadratic-programming}
    \begin{cases}
      w = -\lambda Q^+ (q - t) + (\openone - Q^+Q) v,\\
      \lambda  (q-t)^T Q^+  (q-t) =  (q-t) (\openone  -
      Q^+Q) v - 1,
    \end{cases}
  \end{align}
  where $\lambda$  is a Lagrange multiplier  and $v$ is
  an arbitrary vector.

  Notice that $(q-t)^T Q^+ (q-t)  \ge 0$ since $Q^+ \ge
  0$ and that  $(q-t)^T (\openone - Q^+Q)  (q-t) \ge 0$
  since $(\openone - Q^+Q)$ is a projector.  We need to
  distinguish          four           cases          in
  Eqs.~\eqref{eq:quadratic-programming}.

  {\em First  case}. Let  $(q-t)^T Q^+  (q-t) >  0$ and
  $(q-t)^T (\openone - Q^+Q) (q-t) > 0$. Upon taking
  \begin{align*}
    v   =   \frac{(\openone  -   Q^+Q)   (q-t)}{(q-t)^T
      (\openone - Q^+Q) (q-t)},
  \end{align*}
  one has $\lambda = 0$ and $w = v$.  Therefore $\N{S^T
    w}_2^2 = 0$, namely probability $q$ is incompatible
  with POVM $\pi$.

  {\em Second  case}. Let $(q-t)^T  Q^+ (q-t) =  0$ and
  $(q-t)^T (\openone  - Q^+Q) (q-t) >  0$.  Upon taking
  $v$  again  as  above,  one  has  that  $\lambda$  is
  undetermined and $w = v$.  Therefore $\N{S^T w}_2^2 =
  0$, namely probability $q$  is incompatible with POVM
  $\pi$.

  {\em Third  case}. Let  $(q-t)^T Q^+  (q-t) >  0$ and
  $(q-t)^T (\openone - Q^+Q) (q-t) = 0$. Upon taking $v
  = 0$ one has
  \begin{align*}
    \lambda = - \frac{1}{(q-t)^T Q^+ (q-t)}
  \end{align*}
  and
  \begin{align*}
    w = \frac{Q^+(q-t)}{(q-t)^TQ^+(q-t)}.
  \end{align*}
  Therefore one has
  \begin{align*}
    \N{S^T    w}_2^2     =    \left[    (q-t)^TQ^+(q-t)
    \right]^{-\frac12},
  \end{align*}
  namely probability $q$ is  compatible with POVM $\pi$
  only   if   Eqs.~\eqref{eq:qubit-compatibility}   are
  satisfied.
  
  {\em Fourth  case}. Let $(q-t)^T  Q^+ (q-t) =  0$ and
  $(q-t)^T  (\openone -  Q^+Q) (q-t)  = 0$.   Condition
  $(q-t)^T  Q^+ (q-t)  = 0$  implies $\left|  S^+ (q-t)
  \right|_2 = 0$  and thus $S^+ (q-t) = 0$  and thus $Q
  Q^+ (q-t) =  0$ and thus $Q^+ Q (q-t)  = 0$.  For the
  final  implication we  used  the fact  that from  the
  definition  of  Moore-Penrose pseudoinverse  and  the
  symmetry of $Q$ it follows that  $Q^+ Q = (Q^+ Q)^T =
  Q^T (Q^T)^+ = Q  Q^+$. Condition $(q-t)^T (\openone -
  Q^+Q) (q-t) =  0$ implies $(\openone -  Q^+Q) (q-t) =
  0$ since $\openone - Q^+Q$ is a projector.  Therefore
  altogether they imply the following system
  \begin{align*}
    \begin{cases}
      Q^+Q (q-t) = 0,\\
      (\openone - Q^+Q) (q-t) = 0.
    \end{cases}
  \end{align*}
  This system  in turn  implies $q  = t$  and therefore
  probability $q$  is compatible with POVM  $\pi$ (upon
  input of $\openone/d$).

  Therefore $(q-t)^T  (\openone -  Q^+Q) (q-t) =  0$ is
  necessary for  probability $p$ to be  compatible with
  POVM $\pi$.   Notice also  that $(q-t)^T  (\openone -
  Q^+Q) (q-t)  = 0$ if  and only if $(\openone  - Q^+Q)
  (q-t) =  0$. Therefore probability $p$  is compatible
  with     POVM     $\pi$     if    and     only     if
  Eqs.~\eqref{eq:qubit-compatibility} are satisfied.
\end{proof}

{\em  Applications}   ---  We  have  provided   a  full
characterization   of   $\pds{}{\pi}$   in   terms   of
$\pds{1}{\pi}$     for     {\em    any}     POVM     in
Proposition~\ref{prop:meastest}, and a closed-form full
characterization of $\pds{1}{\pi}$  for {\em any qubit}
POVM  in  Theorem~\ref{thm:qubit-compatibility}. As  an
application, let us now  specify our general results to
the                 depolarized                 version
$\mathcal{D}_\lambda^\dagger(\pi)$  of  any  qubit  SIC
POVM or MUB $\pi := \{ \pi_y \}$.

We   first  recall   that   the  depolarizing   channel
$\mathcal{D}_\lambda$,  modelling  isotropic noise,  is
defined as $\mathcal{D}_\lambda : \rho \to \lambda \rho
+  (1 -  \lambda)  \Tr[\rho] d^{-1}  \openone$ for  any
state $\rho$, and $\mathcal{D}_\lambda^\dagger$ denotes
channel $\mathcal{D}$  in the Heisenberg  picture, i.e.
$\Tr[\mathcal{D}_\lambda(\rho) \; \pi_y]  = \Tr[\rho \;
\mathcal{D}_\lambda^\dagger(\pi_y)]$   for  any   state
$\rho$ and any effect $\pi_y$.

An informationally complete rank-one POVM $\{ \pi_y \}$
such  that  $\braket{\pi_y  |  \pi_y}  =  N_d$  and  $|
\braket{\pi_y |  \pi_{y' \neq  y}} |^2 =  N_d^2 C^2_d$,
for  some  $N_d$ and  $C_d$  that  depend only  on  the
dimension  $d$,  is called a  symmetric,  informationally
complete  (SIC).  By  trivial  computation, it  follows
that    $N_d     =    2(d+1)^{-1}$    and     $C_d    =
(d-1)(d^2+d-2)^{-1}$  for real  SIC POVMs,  and $N_d  =
d^{-1}$ and  $C_d = (d+1)^{-1}$ for  complex SIC POVMs.
In the qubit case, the  only real and complex SIC POVMs
are, up to unitaries  and anti-unitaries, the trine and
tetrahedral POVMs, respectively.

Then,    the    following     result    follows    from
Theorem~\ref{thm:qubit-compatibility}, as  shown in the
Supplemental Material~\cite{supmat}.
\begin{cor}
  An output  probability distribution  $q = \{  q_y \}$
  belongs             to            the             set
  $\pds{1}{\mathcal{D}^\dagger_\lambda(\pi)}$   if  and
  only if
  \begin{align*}
    \N{q}_2^2 \le \frac{\lambda^2 + 2}6,
  \end{align*}
  if $\pi$ is a real SIC, and
  \begin{align*}
    \N{q}_2^2 \le \frac{\lambda^2 + 3}{12},
  \end{align*}
  if $\pi$ is a complex SIC.
\end{cor}

An informationally complete rank-one POVM $\{ \pi_{z,t}
\}$ such that $\{ \ket{\pi_{z,t}} \}$ is an orthonormal
basis for  any $t$,  $\braket{\pi_{z,t} |  \pi_{z,t}} =
N_d$,  and $|  \braket{\pi_{z',t'} |  \pi_{z,t}} |^2  =
N_d^2 C^2_d$ for $t \neq  t'$, for some $N_d$ and $C_d$
that  depend  only  on  the dimension  $d$,  is  called
mutually unbiased basis (MUB).  By trivial computation,
it follows  that $N_d = \floor{d/2+1}^{-1}$  and $C_d =
d^{-1}$ for real MUBs, and  $N_d = (d+1)^{-1}$ and $C_d
= d^{-1}$  for complex  MUBs.  In  the qubit  case, the
only real  and complex  MUBs are,  up to  unitaries and
anti-unitaries,  the   square  and   octahedral  POVMs,
respectively.

Then,    the    following     result    follows    from
Theorem~\ref{thm:qubit-compatibility}, as  shown in the
Supplemental Material~\cite{supmat}.
\begin{cor}
  An output  probability distribution  $q = \{  q_y \}$
  belongs             to            the             set
  $\pds{1}{\mathcal{D}^\dagger_\lambda(\pi)}$   if  and
  only if
  \begin{align*}
    \begin{cases}
      q_{2y} +  q_{2y+1}  =  \frac12, & \quad y = 0, 1,\\
      \N{q}_2^2 \le \frac{\lambda^2 + 2}8, &
      \end{cases}
  \end{align*}
  if $\pi$ is a real MUB,and
  \begin{align*}
    \begin{cases}
      q_{2y} +  q_{2y+1} =  \frac13, & \quad y = 0, 1, 2,\\
      \N{q}_2^2 \le \frac{\lambda^2 + 3}{18}, &
    \end{cases}
  \end{align*}
  if $\pi$ is a complex MUB.
\end{cor}

{\em  Conclusion and  outlooks}  ---  We addressed  the
problem  of  characterizing  the set  $\pds{}{\pi}$  of
input-output  correlations  compatible with  any  given
POVM $\pi$,  upon the input  of any set of  states.  We
provided    as   preliminary    results   a    complete
characterization   of   $\pds{}{\pi}$   in   terms   of
$\pds{1}{\pi}$, i.e.   the range of $\pi$,  only.  This
is in  stark contrast with the  analogous scenario with
respect to quantum channels,  which in general requires
more than one input~\cite{DBB16}.   Our main result was
to conclusively settle the  problem for qubit POVMs, by
deriving   a  full   characterization   of  the   range
$\pds{1}{\pi}$   for  any   given  qubit   POVM  $\pi$,
geometrically interpreted  as an ellipsoid  embedded in
an  $n$-dimensional real  space.   As applications,  we
explicitly discussed the particular cases of qubit real
and  complex SIC  POVMs  and MUBs  in  the presence  of
isotropic noise.  Our results  represent a further step
towards the characterization  of time-like correlations
compatible with  quantum theory~\cite{DBB16}.   In this
sense,    our     results    provide     the    optimal
device-independent test of quantum measurements.

An important problem  left open is that  of providing a
closed-form   full   characterization  of   the   range
$\pds{1}{\pi}$ for POVMs in dimensions higher than two.
An   interesting    related   problem   is    that   of
characterizing the set  of correlations compatible with
a given  family of  states, rather  than a  given POVM.
Finally,   we  would   like  to   mention  a   possible
application of our results in the context of {\em clean
  POVMs}~\cite{BKDPW05}, where it  has been shown that,
for   any  two   POVMs   $\pi$   and  $\Pi$,   whenever
$\pds{1}{\pi}  \subseteq  \pds{1}{\Pi}$, one  has  that
$\pi   =   \mathcal{L}(\Pi)$   for  some   linear   map
$\mathcal{L}$  which  is  positive on  the  support  of
$\Pi$.           The           implications          of
Theorem~\ref{thm:qubit-compatibility}  in this  context
will  be   addressed  by  the  present   authors  in  a
forthcoming work.

Remarkably,   any  observed   input-output  correlation
falsifies some hypothesized POVMs.  Thus, the presented
results  are  particularly  suitable  for  experimental
implementation, for example by using the techniques and
statistical analysis discussed in Refs.~\cite{HGMBAT12,
  ABCB11},  including an  analysis  of  the effects  of
finite statistics.

{\em Acknowledgements} M.~D.  acknowledges support from
the Singapore  Ministry of Education  Academic Research
Fund  Tier  3   (Grant  No.   MOE2012-T3-1-009).   F.~B
acknowledges  support   from  the  JSPS   KAKENHI,  No.
26247016.    V.~V.   acknowledges   support  from   the
Ministry  of Education  and  the  Ministry of  Manpower
(Singapore).

\newpage

\section{Supplemental Material}
\label{sec:supmat}

Here  we prove  those results  reported in  the article
``Device-independent tests of quantum measurements'' by
the present  authors (M.   Dall'Arno, S.   Brandsen, F.
Buscemi, and  V.  Vedral)  whose proofs,  being lengthy
but relatively straightforward,  has only been outlined
in the main text.   The numbering of statements follows
that of the article.

\setcounter{cor}{0}

\begin{cor}
  An output  probability distribution  $q = \{  q_y \}$
  belongs             to            the             set
  $\pds{1}{\mathcal{D}^\dagger_\lambda(\pi)}$   if  and
  only if
  \begin{align}
    \label{eq:realsic}
    \N{q}_2^2 \le \frac{\lambda^2 + 2}6,
  \end{align}
  if $\pi$ is a real SIC, and
  \begin{align}
    \label{eq:complexsic}
    \N{q}_2^2 \le \frac{\lambda^2 + 3}{12},
  \end{align}
  if $\pi$ is a complex SIC.
\end{cor}

\begin{proof}
  Let us first prove the  real case.  The POVM $\pi_y =
  \ket{\pi_y}\bra{\pi_y}$   with  $\ket{\pi_y}   =  U^y
  \ket{0}$  where $U  := e^{-i  \frac{\pi}3 Y}$  is the
  unique (up to unitaries  and anti-unitaries) real SIC
  POVM of a qubit. Thus, $t$ and $Q$ are given by
  \begin{align*}
    t = \frac13
    \begin{pmatrix}
      1 \\
      1 \\
      1
    \end{pmatrix},
    \qquad
    Q = \frac{\lambda}{18}
    \begin{pmatrix}
      2 & -1 & -1 \\
      -1 & 2 & -1 \\
      -1 & -1 & 2
    \end{pmatrix}.
  \end{align*}
  
  Let us  first take  $\lambda >  0$.  The  equality in
  Eqs.~\eqref{eq:qubit-compatibility}           becomes
  $\sum_{y=0}^2   q_y   =   1$.   The   inequality   in
  Eqs.~\eqref{eq:qubit-compatibility} becomes
  \begin{align*}
    (q-t)^{T}  Q^{+}  (q-t)  =  \lambda^{-2}  \left(  6
      \sum_{y=0}^2 q_y^2 - 2 \right) \le 1,
  \end{align*}
  through use of  the identity $\sum_{y < z}  q_y q_z =
  \frac{1}{2} (1 - \sum_y q_y^2)$.

  Let  us now  take  $\lambda =  0$.   The equality  in
  Eqs.~\eqref{eq:qubit-compatibility} becomes  $q = t$.
  The inequality in Eqs.~\eqref{eq:qubit-compatibility}
  is  trivially satisfied.  Then Eq.~\eqref{eq:realsic}
  immediately follows.

  Let us then prove the complex case. The POVM $\pi_y =
  \frac12 \ket{\pi_y}\bra{\pi_y}$ with
  \begin{align*}
    \ket{\pi_y}     :=     \sigma_y     \left(     \cos
      \frac{\arctan{\sqrt2}}2         \ket{0}         +
      e^{i\frac{\pi}4}   \sin   \frac{\arctan{\sqrt2}}2
      \ket{1} \right),
  \end{align*}
  where  $\sigma  := (\openone_2,  \sigma_X,  \sigma_Y,
  \sigma_Z)$ are the Pauli  matrices, is the unique (up
  to unitaries and anti-unitaries) SIC POVM of a qubit.
  Thus, $t$ and $Q$ are given by
  \begin{align*}
    t = \frac14
    \begin{pmatrix}
      1 \\
      1 \\
      1 \\
      1
    \end{pmatrix},
    \qquad
    Q = \frac{\lambda}{48}
    \begin{pmatrix}
      3 & -1 & -1 & -1 \\
      -1 & 3 & -1 & -1 \\
      -1 & -1 & 3 & -1 \\
      -1 & -1 & -1 & 3
    \end{pmatrix}.
  \end{align*}

  Let us  first take  $\lambda >  0$.  The  equality in
  Eqs.~\eqref{eq:qubit-compatibility}           becomes
  $\sum_{y=0}^3   q_y   =   1$.   The   inequality   in
  Eqs.~\eqref{eq:qubit-compatibility} becomes
  \begin{align*}
    (q-t)^{T}  Q^+  (q-t)   =  \lambda^{-2}  \left(  12
      \sum_{y=0}^2 q_y^2 - 3 \right) \le 1,
  \end{align*}
  through use of  the identity $\sum_{y < z}  q_y q_z =
  \frac{1}{2} (1 - \sum_y q_y^2)$.

  Let  us now  take  $\lambda =  0$.   The equality  in
  Eqs.~\eqref{eq:qubit-compatibility} becomes  $q = t$.
  The inequality in Eqs.~\eqref{eq:qubit-compatibility}
  is         trivially         satisfied.          Then
  Eq.~\eqref{eq:complexsic} immediately follows.
\end{proof}

\begin{cor}
  An output probability distribution $q  = \{ q_y \}$ b
  belongs             to            the             set
  $\pds{1}{\mathcal{D}^\dagger_\lambda(\pi)}$   if  and
  only if
  \begin{align}
    \label{eq:realmub}
    \begin{cases}
      q_{2y} +  q_{2y+1}  =  \frac12, & \quad y = 0, 1,\\
      \N{q}_2^2 \le \frac{\lambda^2 + 2}8, &
      \end{cases}
  \end{align}
  if $\pi$ is a real MUB,and
  \begin{align}
    \label{eq:complexmub}
    \begin{cases}
      q_{2y} +  q_{2y+1} =  \frac13, & \quad y = 0, 1, 2,\\
      \N{q}_2^2 \le \frac{\lambda^2 + 3}{18}, &
    \end{cases}
  \end{align}
  if $\pi$ is a complex MUB.
\end{cor}

\begin{proof}
  Let us first prove the  real case.  The POVM $\pi_y =
  \ket{\pi_y}\bra{\pi_y}$         with        $\sigma_k
  \ket{\pi_{2k,2k+1}}  = \pm  \ket{\pi_{2k,2k+1}}$ with
  $k  =  0, 1$  is  the  unique  (up to  unitaries  and
  anti-unitaries) real SIC POVM  of a qubit.  Thus, $t$
  and $Q$ are given by
  \begin{align*}
    t = \frac14
    \begin{pmatrix}
      1 \\
      1 \\
      1 \\
      1
    \end{pmatrix},
    \qquad
    Q = \frac1{16}
    \begin{pmatrix}
      1 & -1 & 0 & 0 \\
      -1 & 1 & 0 & 0 \\
      0 & 0 & 1 & -1 \\
      0 & 0 & -1 & 1
    \end{pmatrix}.
  \end{align*}

  Let us  first take  $\lambda >  0$.  The  equality in
  Eqs.~\eqref{eq:qubit-compatibility}  becomes  $q_0  +
  q_1  =  q_2   +  q_3  =  1/2$.    The  inequality  in
  Eqs.~\eqref{eq:qubit-compatibility} becomes
  \begin{align*}
    (q-t)^{T}  Q^+   (q-t)  =  \lambda^{-2}   \left(  8
      \sum_{y=0}^2 q_y^2 - 2 \right) \le 1,
  \end{align*}
  through   use  of   the  identity   $q_{0}  q_{1}   =
  -\frac{1}{2}(q_{0}^{2}+q_{1}^{2})+\frac{1}{8}$.

  Let  us now  take  $\lambda =  0$.   The equality  in
  Eqs.~\eqref{eq:qubit-compatibility} becomes  $q = t$.
  The inequality in Eqs.~\eqref{eq:qubit-compatibility}
  is trivially  satisfied.  Then Eq.~\eqref{eq:realmub}
  immediately follows.

  Let us then prove the complex case. The POVM $\pi_y =
  \ket{\pi_y}\bra{\pi_y}$         with        $\sigma_k
  \ket{\pi_{2k,2k+1}}  = \pm  \ket{\pi_{2k,2k+1}}$ with
  $k =  0, 1,  2$ is  the unique  (up to  unitaries and
  anti-unitaries) real SIC POVM  of a qubit.  Thus, $t$
  and $Q$ are given by
  \begin{align*}
    t = \frac16
    \begin{pmatrix}
      1 \\
      1 \\
      1 \\
      1 \\
      1 \\
      1
    \end{pmatrix}.
    \qquad
    Q = \frac1{36}
    \begin{pmatrix}
      1 & -1 & 0 & 0 & 0 & 0 \\
      -1 & 1 & 0 & 0 & 0 & 0\\
      0 & 0 & 1 & -1 & 0 & 0 \\
      0 & 0 & -1 & 1 & 0 & 0 \\
      0 & 0 & 0 & 0 & 1 & -1 \\
      0 & 0 & 0 & 0 & -1 & 1
    \end{pmatrix}.
  \end{align*}

  Let  us now  take  $\lambda >  0$.   The equality  in
  Eqs.~\eqref{eq:qubit-compatibility}  becomes  $q_0  +
  q_1 = q_2  + q_3 = q_4 + q_5  = 1/3$.  The inequality
  in Eqs.~\eqref{eq:qubit-compatibility} becomes
  \begin{align*}
    (q-t)^{T}  Q^+  (q-t)   =  \lambda^{-2}  \left(  18
      \sum_{y=0}^2 q_y^2 - 3 \right) \le 1,
  \end{align*}
  through   use  of   the  identity   $q_{0}  q_{1}   =
  -\frac{1}{2}(q_{0}^{2}+q_{1}^{2})+\frac{1}{18}$.

  Let  us now  take  $\lambda =  0$.   The equality  in
  Eqs.~\eqref{eq:qubit-compatibility} becomes  $q = t$.
  The inequality in Eqs.~\eqref{eq:qubit-compatibility}
  is         trivially          satisfied.         Then
  Eq.~\eqref{eq:complexmub} immediately follows.
\end{proof}

\begin{prop}
  \label{prop:corrmatrix}
  For any  qubit POVM $\{  \pi_y \}$, one has that 
  \begin{align*}
    -\frac14 \le Q_{y_0,  y_1} := \frac12 \Tr[\pi_{y_0}
    \pi_{y_1}] -  \frac14 \Tr[\pi_{y_0}] \Tr[\pi_{y_1}]
    \le \frac14,
  \end{align*}
  where  the bounds  are saturated  if $\pi_{y_0}$  and
  $\pi_{y_1}$ are rank-one, unit-trace, orthogonal (for
  the lower bound) or  coincident (for the upper bound)
  effects.
\end{prop}

\begin{proof}
  The problem  consists in maximizing  (minimizing) the
  quantity $2 \Tr[A  B] - \Tr[A] \Tr[B]$ over  $0 \le A
  \le  \openone$, $0  \le  B \le  \openone$, under  the
  constraint that either $A  = B$ (diagonal elements of
  $Q$), or $A +  B \le \openone$ (off-diagonal elements
  of $Q$).  In the  following we relax this constraint,
  at the  end we will  show that our  solutions satisfy
  it.
  
  For any given $A$ with  eigenvalues $a_0 \ge a_1$ and
  corresponding eigenvectors $\phi_0$ and $\phi_1$, one
  has $A = (a_0 - a_1) \phi_0 + a_1 \openone$, hence
  \begin{align*}
    2  \Tr[A  B]   -  \Tr[A]  \Tr[B]  =   (a_0  -  a_1)
    \Tr[(\phi_0 - \phi_1) B].
  \end{align*}
  Therefore, the  maximum (minimum,  respectively) over
  $B$  for fixed  $A$  is  $a_0 -  a_1$  ($a_1 -  a_0$,
  respectively)  attained  when  $B  =  \phi_0$  ($B  =
  \phi_1$, respectively).
  
  Now, the maximum (minimum,  respectively) over $A$ of
  $a_0 - a_1$ ($a_1 - a_0$, respectively) is $1$ ($-1$,
  respectively),  attained  when  $A$  is  a  rank-one,
  unit-trace effect.   Hence, the maximum over  $A$ and
  $B$  is  attained  when  $A$ and  $B$  are  rank-one,
  unit-trace,  coincident effects,  that  is  $A =  B$,
  while the minimum  over $A$ and $B$  is attained when
  $A$  and  $B$  are rank-one,  unit-trace,  orthogonal
  effects,  that  is $A  +  B  = \openone$.  Thus,  the
  constraint is verified.
\end{proof}
\end{document}